\documentclass[12pt,letterpaper]{article}

\usepackage[T1]{fontenc}
\usepackage{lmodern}
\usepackage[margin=1in]{geometry}
\usepackage{amsmath,amssymb,amsthm,mathtools}
\usepackage{array,booktabs,tabularx}
\usepackage[round,authoryear]{natbib}
\usepackage[expansion=false]{microtype}
\usepackage[hidelinks]{hyperref}
\hypersetup{
  pdftitle={Task Bundling and Learning about Competence},
  pdfauthor={Ekaterina Korotkova and Georgy Lukyanov},
  pdfkeywords={task assignment, coarse information, competence, Blackwell order, auditing}
}

\AddToHook{cmd/appendix/after}{}

\newtheorem{lemma}{Lemma}
\newtheorem{proposition}{Proposition}
\newtheorem{corollary}{Corollary}
\newtheorem{observation}{Observation}

\newcommand{\Bundled}{\mathsf{B}}
\newcommand{\Narrow}{\mathsf{N}}
\newcommand{\Separate}{\mathsf{F}}
\newcommand{\E}{\mathbb{E}}

\title{Task Bundling and Learning about Competence\thanks{We thank the editor, Sean Horan, and the two anonymous referees at \emph{Mathematical Social Sciences} for their comments. We also thank Harry Pei and participants in the ICEF research seminar for helpful discussions, and an anonymous referee on an earlier version for suggesting the fixed-workload comparison. This research received financial support from the French National Research Agency (ANR) under grant ANR-17-EURE-0010 (Investissements d'Avenir program). All remaining errors are our own.}}

\author{Ekaterina Korotkova\thanks{International College of Economics and Finance, National Research University Higher School of Economics, 11 Pokrovsky Boulevard, Moscow, 109028, Russia.}
  \and
  Georgy Lukyanov\thanks{Toulouse School of Economics, 1 Esplanade de l'Universit\'e, Toulouse, 31080, France.
  Corresponding author. Email address: \href{mailto:georgy.lukyanov@tse-fr.eu}{\nolinkurl{georgy.lukyanov@tse-fr.eu}} (Georgy Lukyanov).}}
\date{}

\begin{document}

\maketitle

\begin{abstract}
A buyer assessing an expert of high or low competence observes only whether each project succeeds. We compare giving the expert all the tasks in one project with assigning the same workload across separate projects, each with an outside input. At positive, imperfect outside reliability, success in every separate project gives the same posterior as one bundled success and occurs less often. The bundled project Blackwell dominates up to a cutoff in outside reliability. Above it the reports are incomparable until the provider is perfectly reliable, when separate projects dominate. With a third possible competence level, the two unaudited assignments are incomparable at every positive, imperfect reliability. A buyer with a strictly convex continuation value still has her own cutoff. In a contractor trial, changing only the accuracy of the incumbent supplier can reverse which assignment a buyer prefers before deciding whether to renew. For a single two-task project we also derive the audit rate at which assigning the expert both tasks dominates assigning one.
\end{abstract}

\begin{quotation}
\end{quotation}

\section{Introduction}\label{sec:intro}

A buyer is testing a contractor before awarding her next order. She has two records for the contractor to convert. She can put them in one order or pair each with a fresh record converted by an outside provider. The qualification service reports only whether both conversions in an order are correct. When an order fails, how much of the blame should fall on the contractor?

Suppose the contractor uses one of two versions of a routine. The buyer knows their accuracies and wants to learn which version will process her next record. With a perfect outside provider, separate orders reveal which conversions the contractor got right. Bundling conceals the difference between one mistake and two. An imperfect provider can hide a correct conversion behind its own mistake.

Consider $n$ tasks. At contractor accuracy $q$, a bundled order passes with probability $q^n$. If the expert performs one task per order and the independent outside component succeeds with probability $p$, all $n$ orders pass with probability $p^nq^n$. For $0<p<1$ these successes give the same posterior as a bundled pass but occur less often. With two possible accuracies the bundled verdict can reproduce the separate reports whenever $p$ is sufficiently low.

Why do two accuracies matter? A rule using only a bundled pass or failure must report a pass on one separate order with probability $a q^n+b(1-q^n)$, for fixed reporting probabilities $a$ and $b$. The required probability is $pq$. As a function of $q^n$, it lies on a strictly concave curve when $p>0$. A line can meet that curve at two accuracies but cannot meet it at three. If a bundled verdict cannot reproduce even one separate report, it cannot reproduce all of them.

Consider a buyer whose incumbent is as reliable as the contractor is on average. A bundled pass persuades her to renew and a failure sends her back to the incumbent. This remains true with three possible versions. When $p$ is close to zero, separate orders almost always fail and are worth less to her at equal fees.

Our two-version example asks how buyers choose when neither trial dominates. At the same trial fee, a buyer with a dependable incumbent bundles because she needs exceptionally good news to renew. A buyer with a less dependable incumbent separates because two failures would change her mind about renewing. They can disagree even when separate orders cost more.

{\clubpenalty=10000 With two competence levels the cutoff for bundled dominance is higher if the alternative is a single narrow project, in which the expert does only one task. Holding the expert's workload fixed gives the buyer repeated opportunities to observe failures. For any finite set of distinct accuracies with positive prior weights, a strictly convex continuation value yields a decision-specific cutoff. We also study a single two-task project under a common random audit.\par}

\section{Related literature}\label{sec:literature}

\citet{Meyer1994} also holds workers' time budgets fixed. When two juniors divide their time equally between two teams, their combined contribution cancels from the difference in team outputs. The firm learns more about which senior is better and gives up the information about relative junior ability available when each junior works on one team. For a firm promoting one senior in steady state, junior sharing is optimal when prior uncertainty about ability is large relative to production noise. Assigning each junior to a single team is optimal when it is small. One of these arrangements attains the optimum even with a free choice of time allocations. Our buyer has a single contractor to assess. Spreading its work across projects exposes each verdict to an outside failure. This changes how often she gets the evidence she needs to renew.

\citet{Meyer1991} studies a firm that learns only who won each of a series of contests. In the last contest, the firm can give the leader a head start that reflects the earlier rankings. The winner then identifies the worker the firm would promote if it observed the final output difference. Our buyer chooses which tasks go together before seeing any results, and the acceptance test is fixed. Here too, seeing a finer report can leave a decision unchanged. With a perfect outside provider, a buyer who renews only after two correct conversions has no use for the distinction between one mistake and two.

\citet{DanaSpier2018} study a related difficulty in attributing a bad outcome. A consumer who bought one good from a branded firm and another from the competitive fringe may blame the fringe and continue to buy the branded good. If the firm sells both goods together, a bad signal can cost it both sales. The prospect of losing those sales can sustain high quality. In our setting the parties' task accuracies are given. Attribution matters for the buyer's choice of supplier on the next order.

A firm can also learn by assigning workers to jobs over time \citep{MacDonald1982,Pastorino2024} or by rotating them across activities \citep{Ortega2001}. Career concerns give workers a reason to change their effort with the scope of the job and the information it conveys \citep{Dewatripont1999,Dewatripont1999b,BarIsaacHorner2014,AshworthBuenoFriedenberg2017}. Related questions arise in allocating political responsibility when task outcomes are observable \citep{AshworthBueno2017}, and in aggregating performance measures for incentive pay \citep{SabacYoo2018}.

We use the order of \citet{Blackwell1953} to compare what the buyer can learn. With the bundled verdict as the source, the geometric criterion of \citet{Wu2023} determines which reports it can reproduce. In the reverse direction we also need the probability of the most favorable separate report.

\section{Task assignment and project outcomes}\label{sec:model}

There are $n\geq2$ complementary tasks in a project. The buyer assigns them before any outcomes are observed. In our benchmark the expert is of type $\theta\in\{H,L\}$, with prior probability $\mu\in(0,1)$ of type $H$. The type is fixed throughout. Conditional on type, the expert completes each assigned task correctly and independently with probability
\[
q_\theta=\begin{cases}h,&\theta=H,\\ \ell,&\theta=L,\end{cases}
\qquad 0<\ell<h<1.
\]
The buyer acts on public information, and the parties cannot alter their accuracies in response to her assignment. Hence the assignment itself conveys no information about competence. A project succeeds if and only if all its tasks succeed. The buyer sees one success or failure per project. Section~\ref{sec:audit} allows her to see the individual outcomes when a project is audited.

If the expert performs just one task, an outside provider performs the rest. Write $p\in[0,1]$ for the known probability that the provider completes \emph{all} of them correctly. Its success is independent of the expert's type and outcomes, and outside outcomes are independent across projects. For a project of a given size it suffices to know $p$. When we vary $n$, we also consider $n-1$ independent outside tasks with common success probability $r$, so $p=r^{n-1}$. Keeping $0<p<1$ fixed as $n$ grows then requires the individual tasks to become more reliable, with $r=p^{1/(n-1)}$ tending to one.

Table~\ref{tab:architectures} gives the three assignments. Here $q$ denotes the expert's accuracy, $Y$ is one project's verdict, and $K$ counts successes across separate projects.

\begin{table}[htbp]
\centering
\caption{Who performs the tasks and what the buyer observes}
\label{tab:architectures}
\small
\setlength{\tabcolsep}{4pt}
\renewcommand{\arraystretch}{1.12}
\begin{tabularx}{\linewidth}{@{}
>{\raggedright\arraybackslash}p{0.235\linewidth}
>{\raggedright\arraybackslash}p{0.17\linewidth}
>{\raggedright\arraybackslash}p{0.205\linewidth}
>{\raggedright\arraybackslash}X@{}}
\toprule
Assignment & Expert's tasks & Outside input & Reports given $q$ \\
\midrule
Bundled project ($\Bundled_n$) & All $n$ in one project & None &
One report\newline $Y\sim\mathrm{Bernoulli}(q^n)$ \\
\addlinespace[3pt]
Narrow project ($\Narrow_p$) & One task in one project & One complement &
One report\newline $Y\sim\mathrm{Bernoulli}(pq)$ \\
\addlinespace[3pt]
Separate projects ($\Separate_{n,p}$) & One in each of $n$ projects & One complement per project &
$n$ reports\newline $K\sim\mathrm{Binomial}(n,pq)$ \\
\bottomrule
\end{tabularx}
\end{table}

In moving from the bundled to the narrow project the buyer gives $n-1$ of the expert's tasks to the provider. In the fixed-workload comparison the expert does $n$ tasks under either assignment. Each of the $n$ separate projects also needs its own outside component, for a total of $n(n-1)$ outside tasks. Given $q$, each project succeeds independently with probability $pq$. Conditional on any count $K=k$ of positive probability, all $\binom nk$ ways of obtaining those successes are equally likely. The buyer therefore learns as much about competence from the count as from the labeled verdicts.

If the provider's reliability is itself drawn independently for each project and the draws are unobserved, averaging over each draw gives $p$. Suppose instead that an unobserved reliability $P$, independent of $q$, applies to every project in the trial. If verdicts $Y_i$ are independent with success probability $Pq$ conditional on $(P,q)$, then for $i\ne j$,
\[
\operatorname{Cov}(Y_i,Y_j\mid q)=q^2\operatorname{Var}(P).
\]
The count remains sufficient, but a common shock with positive variance gives a mixture of binomial laws. Replacing $P$ by its mean would give the wrong probability of a run of failures. Unequal task accuracies can make the identity of a failed project informative. Correlation among the expert's own outcomes can likewise change the likelihood ratios.

Let $M_A$ be the buyer's posterior after assignment $A$, and $V(m)$ her expected payoff when she acts optimally at belief $m$. Maximizing over actions makes $V$ convex. Before the reports arrive her expected continuation payoff is
\[
J_A(V)=\E[V(M_A)].
\]
In the sense of \citet{Blackwell1953}, $A$ dominates $A'$ if it yields at least as high a payoff for every continuation decision and every prior. Equivalently, one can reproduce the reports from $A'$ in every state by applying a state-independent randomized rule to the reports from $A$. We call such a rule a garbling. Dominance is strict if there is no garbling in the reverse direction.

Fix a full-support prior. If the source posteriors are affinely independent, a garbling exists exactly when their convex hull contains every target posterior \citep{Wu2023}. A bundled project has two distinct posteriors. In the two-state model they bound an interval, and posterior odds equal prior odds times the likelihood ratio. We can thus test which reports a bundled verdict can reproduce by comparing likelihood ratios (Lemma~\ref{lem:binary}). Separate projects have $n+1$ distinct posteriors on a line when $p>0$. The sufficiency part of the criterion requires affine independence and does not apply in this direction.

\section{When does bundling improve learning?}\label{sec:ranking}

\subsection{A one-project benchmark}\label{sec:scope}

Consider a patient choosing in advance whether the doctor who will diagnose her condition will also treat it. Suppose recovery requires both a correct diagnosis and an appropriate treatment, and the patient observes only whether she recovers. If she obtains treatment elsewhere, a failure may leave her little reason to doubt the diagnosis. These are the bundled and narrow projects in the case of two tasks.

For $n\geq2$, a bundled failure gives the likelihood ratio
\[
d_n=\frac{1-h^n}{1-\ell^n}.
\]
Define
\begin{equation}
\bar p_n=\frac{1-d_n}{h-\ell d_n}.
\label{eq:one-threshold}
\end{equation}
This cutoff is strictly between zero and one. For two tasks it is
\[
\bar p_2=\frac{h+\ell}{1+h\ell}.
\]
\begin{proposition}\label{prop:one}
With two competence levels and $n\geq2$, the bundled project Blackwell dominates the narrow project if and only if $p\leq\bar p_n$. For $p>\bar p_n$ the two projects are Blackwell incomparable. The narrow project never dominates the bundled project.
\end{proposition}

At $p=0$ the narrow project always fails. At $p>0$ the likelihood ratio after success is $h/\ell$ in the narrow project and $(h/\ell)^n$ in the bundled project. Putting all the tasks in the expert's hands therefore gives the buyer a more favorable inference from a success. What she infers from a narrow failure depends on
\[
t(p)=\frac{1-ph}{1-p\ell}.
\]
As the provider becomes more reliable, $t(p)$ falls. Below $\bar p_n$ a narrow failure gives the buyer less reason to doubt the expert than a bundled failure. At the cutoff the two failures give the same posterior. Above it, a bundled success is the more favorable report and a narrow failure the more adverse one.

The assignment that makes the order most likely to pass can therefore be the less informative one. Take $n=2$ and fix $h<p\leq\bar p_2$. Giving one task to the provider raises the probability of success from $q^2$ to $pq$ for either type, while the buyer learns more from the bundled project in the Blackwell sense. Now keep the narrow project and improve the provider from $p_1$ to $p_2>p_1$. A report at the new reliability can reproduce the old one by retaining each success with probability $p_1/p_2$ and recording a failure otherwise. Thus increasing $p$ improves information within the narrow project. The same argument applies independently to the verdicts from separate projects.

\subsection{Separate projects at a fixed workload}\label{sec:workload}

Suppose the buyer can try the expert on $n$ tasks however she divides the orders. Giving the expert one task in each of $n$ separate projects exposes each verdict to an independent outside failure. Put
\begin{equation}
s_n=d_n^{1/n}
\qquad\text{and}\qquad
\widehat p_n=\frac{1-s_n}{h-\ell s_n}.
\label{eq:fixed-threshold}
\end{equation}
For $p>0$ the likelihood ratio after $k$ successful projects is
\begin{equation}
\left(\frac{h}{\ell}\right)^k
\left(\frac{1-ph}{1-p\ell}\right)^{n-k},
\label{eq:fragment-lr}
\end{equation}
which increases strictly with $k$. At $p=0$ every separate project fails with certainty.

\begin{proposition}\label{prop:fixed}
Suppose competence has two levels and the expert performs $n\geq2$ tasks under either assignment. The bundled project Blackwell dominates the $n$ separate projects if and only if $p\leq\widehat p_n$. For $\widehat p_n<p<1$ the two assignments are Blackwell incomparable. At $p=1$ separate projects Blackwell dominate the bundled project.
\end{proposition}

Consider the buyer who needs very favorable evidence before she will renew the contract. At $0<p<1$ she can be just as confident after $n$ separate successes as after a bundled success. Only the count $K=n$ gives her that confidence. At accuracy $q$ its probability is $p^nq^n$, while the bundled project passes with probability $q^n$. A reporting rule could increase the frequency of passes by drawing on lower counts, but a pass would then be weaker evidence in favor of renewal. This rules out a garbling from separate projects to the bundled project. At $p=1$ the buyer can simply report a bundled success whenever every separate project succeeds.

For the opposite comparison, consider the buyer's grounds for turning the contractor down. The worst news from separate projects is a run of $n$ failures, with likelihood ratio $t(p)^n$. As long as $t(p)^n\geq d_n$, the bundled posteriors enclose all the separate ones. Lemma~\ref{lem:binary} then gives a garbling in this direction. At $p=\widehat p_n$ the two sets of posteriors have the same convex hull at every interior prior. Bundled dominance is still strict because the all-success count remains too rare to reproduce a bundled success.

\begin{corollary}\label{cor:thresholds}
With two competence levels, the cutoffs for the whole outside component satisfy $0<\widehat p_n<\bar p_n<1$ for every $n\geq2$. The fixed-workload cutoff $\widehat p_n$ decreases strictly with $n$. As $n\to\infty$,
\[
\widehat p_n\sim\frac{h^n}{n(h-\ell)},
\qquad
\bar p_n\sim\frac{h^n}{h-\ell}.
\]
If each of the $n-1$ outside tasks succeeds independently with probability $r$, the fixed-workload and one-project cutoffs in per-task reliability satisfy
\[
\widehat r_n=(\widehat p_n)^{1/(n-1)}<h<
\bar r_n=(\bar p_n)^{1/(n-1)}.
\]
Both per-task cutoffs converge to $h$.
\end{corollary}

For $\widehat p_n<p\leq\bar p_n$ the bundled project dominates the narrow project and is incomparable with separate projects. The opportunity to observe repeated failures is what shrinks the region of dominance when the expert does the same amount of work. At $\bar p_n<p<1$ both comparisons give incomparability.

As project size grows, the fixed-workload cutoff is approximately $1/n$ times the one-project cutoff when both refer to the whole outside component. What if the buyer uses the same technology for each outside task as projects grow? If $0<r<h$, the bundled project eventually dominates both alternatives. If $h<r<1$, it is eventually incomparable with both. At $r=h$ it dominates the narrow project and is incomparable with separate projects at every size. Successful reports themselves become rare for fixed $r<1$. At accuracy $q$ a bundled success has probability $q^n$, and the probability of even one separate success is at most $nqr^{n-1}$.

\subsection{A third competence level}\label{sec:boundary}

An anonymous referee suggested allowing a third competence level. Let $Q\subset(0,1)$ contain at least three distinct accuracies, and keep the unaudited reports in Table~\ref{tab:architectures}. A reporting rule must now work at every $q\in Q$.

\begin{observation}\label{obs:boundary}
For $n\geq2$, the bundled and narrow projects are Blackwell incomparable for every $p>0$. The bundled project and separate projects are incomparable for $0<p<1$. At $p=0$ the bundled project strictly dominates both alternatives. At $p=1$ separate projects strictly dominate the bundled project.
\end{observation}

\begin{proof}
Suppose the buyer reports a narrow success with probability $a$ after a bundled success and with probability $b$ after a bundled failure. To reproduce the narrow project she needs
\[
pq=b+(a-b)q^n \qquad\text{for every }q\in Q.
\]
For $p>0$, plot $pq$ against $x=q^n$. The target probability $p x^{1/n}$ is strictly concave, so no line passes through three of its points. Reversing the roles of the two projects gives the strictly convex function $(x/p)^n$ of $x=pq$. This excludes both garblings.

Selecting one of the separate reports at random gives a narrow success with probability $k/n$ conditional on $K=k$, hence with probability $pq$ at accuracy $q$. Any garbling from the bundled project to separate projects would thus reproduce the narrow project. For the reverse direction, take any $q_i<q_j$ in $Q$. Only $K=n$ attains the bundled pass ratio $(q_j/q_i)^n$. For $0<p<1$ its probability under $q_i$ is $p^n q_i^n<q_i^n$, too little to reproduce a bundled success. At $p=0$ both outside-dependent assignments always fail and the bundled verdict is informative. At $p=1$ reporting success exactly at $K=n$ reproduces the bundled project, and the argument above excludes the reverse garbling.
\end{proof}

With two possible versions and equal net trial returns, the buyer can weakly prefer the bundled project at $p\leq\widehat p_n$ whatever decision follows. A third version removes that positive interval of Blackwell dominance. We turn to the order she places after seeing the reports.

\section{Choosing a trial}\label{sec:value}

\subsection{The renewal decision}\label{sec:renewal}

The buyer has decided to run a qualification trial. Only two software releases are compatible with her records. She knows their accuracies, $\ell$ for the older release and $h$ for the newer one. She cannot inspect the installation, and the same release processes the trial records and her next production record if she renews. The qualification service draws fresh records independently from the same population as the production order. It accepts orders in pairs, keeps its reference answers private for later trials, and tells the buyer whether both records match. She cannot submit a record whose answer she already knows.

The contractor converts both records in a bundled order. With separate orders it converts one record per order, and the outside provider converts the other with independent success probability $p$. The test records have no resale value. Separate orders require two extra records and two outside conversions. Start with a single qualification fee covering either trial, including the outside work. We allow the fees to differ below.

A correct conversion on that record pays one and an incorrect conversion pays zero. The buyer can renew with the contractor or hire her incumbent supplier, whose accuracy is $s\in(\ell,h)$. The two charge the same price, which we normalize to zero. If $m$ is the posterior probability of the newer routine, the expected payoff from renewal is $\ell+(h-\ell)m$. The buyer renews for $m>c$, keeps the incumbent for $m<c$, and is indifferent at equality, where
\[
c=\frac{s-\ell}{h-\ell}.
\]
Relative to using the incumbent, she earns $(h-\ell)V_c(m)$, where $V_c(m)=(m-c)_+$ and $x_+=\max\{x,0\}$. We give expected gains $J_A(V_c)$ in units of $h-\ell$. Here $p$ concerns the provider on the test orders. The incumbent's accuracy $s$ enters only when the buyer awards the production order.

For the first buyer we use the example suggested by an anonymous referee. Set $h=4/5$ and $\ell=2/5$ and put equal prior weight on the two versions. A bundled order passes with probability $(h^2+\ell^2)/2=2/5$. The posteriors conditional on a bundled pass and on a bundled failure are, respectively,
\[
\frac{h^2}{h^2+\ell^2}=\frac45,\qquad
\frac{1-h^2}{2-h^2-\ell^2}=\frac3{10}.
\]
The incumbent has accuracy $s=7/10$, so this buyer needs a posterior above $3/4$ to prefer renewal. A bundled pass is enough. With separate projects at $p>0$, two passes give her the same posterior $4/5$ with probability $(2/5)p^2$. One pass gives odds $2(1-ph)/(1-p\ell)<2$, and two failures give still lower odds. Neither report induces renewal. At $p=0$ the two projects always fail and the buyer keeps her prior belief. Thus
\[
J_{\Bundled_2}(V_{3/4})=\frac25\!\left(\frac45-\frac34\right)=\frac1{50},
\qquad
J_{\Separate_{2,p}}(V_{3/4})=\frac{p^2}{50}.
\]
She strictly prefers the bundled project at every $p<1$. At $p=1$ both trials tell her whether the contractor got both records right. Whenever it made a mistake, she keeps her incumbent, so learning how many mistakes it made has no value to her.

Now fix $p=3/4$. The first buyer earns $1/50$ from bundling and $9/800$ from separate projects in normalized units. Give a second buyer the very same trial and prior, but an incumbent with accuracy $s=8/15$. Her renewal threshold is $c=1/3$. A separate order passes with probability $3/5$ under the newer routine and $3/10$ under the older one. Bayes' rule applied to the two binomial laws gives
\begin{center}
\begin{tabular}{@{}lccc@{}}
\toprule
Number of passes $k$ & $0$ & $1$ & $2$ \\
\midrule
Probability of $K=k$ & $13/40$ & $9/20$ & $9/40$ \\
Posterior $\Pr(H\mid K=k)$ & $16/65$ & $8/15$ & $4/5$ \\
\bottomrule
\end{tabular}
\end{center}
This buyer would renew before seeing any report. A bundled failure takes her below $1/3$, as do two separate failures. One separate pass is enough for her to stay with the contractor. Always renewing would give her $\mu-c=1/6$ above the incumbent. By switching on a report with $m<c$ she gains a further $c-m$, so her two normalized payoffs are
\[
\begin{aligned}
J_{\Bundled_2}(V_{1/3})
 &=\frac16+\frac35\!\left(\frac13-\frac3{10}\right)=\frac{14}{75},\\
J_{\Separate_{2,3/4}}(V_{1/3})
 &=\frac16+\frac{13}{40}\!\left(\frac13-\frac{16}{65}\right)=\frac{39}{200}.
\end{aligned}
\]
She prefers separate projects by $1/120$. Two failures are relatively infrequent, but when they occur the case for going back to the incumbent is strong enough to justify separating the orders. The first buyer had been looking for grounds to leave her incumbent. The two strict preferences establish incomparability at $p=3/4$ within the renewal decision itself. At these same parameters a bundled order dominates a single narrow order. The likelihood ratios after a bundled failure, a narrow failure and two separate failures are $3/7$, $4/7$ and $16/49$, respectively.

In \citet{Meyer1994}, sharing juniors changes what the firm can learn about each cohort. Both our buyers learn about the same contractor. They disagree about the trial because different evidence would change the supplier each hires.

Suppose now that the trials differ in fees or yield commercially useful output. Write $\Pi_A-c_A$ for the net current return from assignment $A$, taking account of outside work and coordination costs. Measure the continuation value $V$ in the same monetary units, and let $\delta>0$ discount the next order. The buyer weakly prefers the bundled project if and only if
\begin{equation}
[\Pi_{\Bundled_n}-c_{\Bundled_n}]
-[\Pi_{\Separate_{n,p}}-c_{\Separate_{n,p}}]
\geq
\delta[J_{\Separate_{n,p}}(V)-J_{\Bundled_n}(V)].
\label{eq:organization-choice}
\end{equation}
For the renewal decision, $V(m)=s+(h-\ell)V_c(m)$. With equally valuable test outputs, the first buyer would pay up to $\delta(h-\ell)(1-p^2)/50$ extra to bundle. At $p=3/4$, the first buyer gains $7/800$ in normalized units from bundling and the second gains $1/120$ from separating. Let $D=c_{\Separate_{2,p}}-c_{\Bundled_2}$ be the extra fee for separate orders in monetary units. With the same test revenues and $h-\ell=2/5$, the two buyers strictly prefer opposite trials precisely when
\[
-\frac{7\delta}{2000}<D<\frac{\delta}{300}.
\]
At the lower endpoint the first buyer is indifferent, and at the upper endpoint the second is indifferent. The buyer with the less dependable incumbent will thus pay a positive premium for the two separate verdicts.

\subsection{A cutoff for strictly convex continuation values}\label{sec:convex}

The first buyer is indifferent between trials at $p=1$ because several reports lead her to take the same action. Consider now a buyer who strictly values the distinction between any two different posteriors that would otherwise be pooled. Let $Q=\{q_1,\ldots,q_d\}\subset(0,1)$ contain $d\geq2$ distinct accuracies, let $\pi_i>0$ be the prior probability of $q_i$, and write $m=(m_1,\ldots,m_d)$ for the posterior. As in Section~\ref{sec:scope}, take $0\leq p_1<p_2\leq1$ and keep each observed success at $p_2$ with probability $p_1/p_2$. From a count of $k$ we thus draw a binomial count with parameters $k$ and $p_1/p_2$. The resulting reports have the law of separate projects at $p_1$ at every accuracy.

\begin{proposition}\label{prop:value}
Fix $n\geq2$, a full-support prior on $Q$, and a finite, strictly convex continuation value $V$ on the posterior simplex. There is a unique $p_V\in(0,1)$ at which the bundled project and separate projects give the same expected continuation payoff. The bundled project gives the higher continuation payoff for $p<p_V$ and separate projects for $p>p_V$. With two competence levels, $p_V>\widehat p_n$. With three or more, the infimum of $p_V$ over such continuation values is zero, even with the prior fixed.
\end{proposition}

One such buyer forecasts which routine is installed under squared-error loss on the vector of type indicators. Her optimal forecast is $m$ and her payoff, up to a constant, is $\sum_i m_i^2$. At $p=0$ only the bundled verdict moves her forecast. At $p=1$ separate projects distinguish counts that a bundled failure pools. In between, thinning successes pools distinct posteriors with positive probability. A strictly convex value makes this pooling costly. The expected payoff from separate projects therefore increases strictly with $p$ and crosses the bundled payoff once.

For any convex value, the posterior after thinning is a conditional average of the original posteriors, so an improvement in the provider weakly benefits the buyer. With finitely many actions the value is piecewise linear. In the two-type renewal example at $p=1$, for instance, a bundled failure pools counts $0$ through $n-1$. Let $m_k=\Pr(H\mid K=k)$ be the separate posteriors, increasing in $k$. Observing the count strictly improves the renewal decision exactly when $m_0<c<m_{n-1}$. Some of the pooled counts then favor renewal and others favor the incumbent. At either endpoint the buyer can still take a common optimal action throughout the pool.

With three or more routines the buyer can still use a bundled verdict to decide whether to renew. Take an incumbent whose accuracy is $s=\E_\pi[q]$, so the buyer is indifferent before the trial. A bundled pass has probability $\E_\pi[q^n]$ and raises expected accuracy to $\E_\pi[q^{n+1}]/\E_\pi[q^n]$. The increase is positive because $q$ and $q^n$ are strictly increasing together on the support. A failure lowers expected accuracy, since posterior expectations average back to $s$. Renewing only after a pass therefore earns $\E_\pi[q^{n+1}]-s\E_\pi[q^n]=\operatorname{Cov}_\pi(q,q^n)>0$ over using the incumbent, measured in production-order units. At $p=0$ separate orders leave the buyer indifferent. Their value is continuous in $p$, so at equal trial fees she strictly prefers bundling for all sufficiently small $p$. Other decisions can give very different cutoffs. The proof of Proposition~\ref{prop:value} constructs a forecast that changes after exactly $n-1$ separate successes and stays the same after either bundled verdict.

\subsection{Quadratic prediction loss}\label{sec:quadratic}

To put a value on the two sources of lost information, return to two types and suppose the buyer predicts the indicator of high competence. Under quadratic loss she would predict $\mu$ before the trial and $M_A$ afterwards. The reduction in expected squared error is
\[
I_A(\mu)=\E[(M_A-\mu)^2].
\]
If $m_B=\mu h^n+(1-\mu)\ell^n$ denotes the probability of a bundled success, Bayes' rule gives
\begin{equation}
I_{\Bundled_n}(\mu)
=\mu^2(1-\mu)^2
\frac{(h^n-\ell^n)^2}{m_B(1-m_B)}.
\label{eq:quadratic-bundle}
\end{equation}
For separate projects the count probabilities conditional on $q$ and under the prior are, respectively,
\[
f_q(k;p)=\binom{n}{k}(pq)^k(1-pq)^{n-k},
\qquad
w_k(p)=\mu f_h(k;p)+(1-\mu)f_\ell(k;p).
\]
Summing over the possible counts gives
\begin{equation}
I_{\Separate_{n,p}}(\mu)
=\mu^2(1-\mu)^2
\sum_{k=0}^{n}
\frac{[f_h(k;p)-f_\ell(k;p)]^2}{w_k(p)}.
\label{eq:quadratic-fragment}
\end{equation}
A count of zero probability contributes zero by continuity. Both losses can now be measured against separate reports with a perfect provider. Reporting success only at $K=n$ produces the bundled verdict. Independently retaining each success with probability $p$ produces the separate reports with the imperfect provider. Suppressing the common prior, we obtain
\[
I_{\Bundled_n}-I_{\Separate_{n,p}}
=
\underbrace{(I_{\Separate_{n,1}}-I_{\Separate_{n,p}})}_{\text{loss from outside failures}}
-
\underbrace{(I_{\Separate_{n,1}}-I_{\Bundled_n})}_{\text{loss from aggregation}}.
\]
As $p$ rises to one, the first loss falls strictly to zero. The second is the cost of learning only that a project failed when the separate reports would reveal how many tasks failed. It is positive and independent of $p$. The two losses are equal at the forecaster's cutoff $p_Q(\mu)\in(\widehat p_n,1)$ from Proposition~\ref{prop:value}.

For the trial in Section~\ref{sec:renewal}, equations \eqref{eq:quadratic-bundle} and \eqref{eq:quadratic-fragment} give $I_{\Bundled_2}(1/2)=3/50$, $I_{\Separate_{2,3/4}}(1/2)=271/6500$, and $I_{\Separate_{2,1}}(1/2)=9/125$. Relative to separate reports from a perfect provider, outside failures raise expected squared prediction error by $197/6500$ and aggregation raises it by $3/250$. At $p=3/4$ the forecaster therefore gains $119/6500$ by bundling.

\begin{table}[htbp]
\centering
\caption{What the same trial is worth for three decisions}
\label{tab:pilot-values}
\small
\begin{tabular}{@{}lccc@{}}
\toprule
Use of the report & Bundled project & Separate projects & Preferred trial\\
\midrule
Renewal, $s=7/10$ & $1/50$ & $9/800$ & Bundled project\\
Renewal, $s=8/15$ & $1/50$ & $17/600$ & Separate projects\\
Quadratic forecast & $3/50$ & $271/6500$ & Bundled project\\
\bottomrule
\end{tabular}

\smallskip
\begin{minipage}{0.94\linewidth}
\footnotesize\emph{Notes.} $n=2$, $h=4/5$, $\ell=2/5$, $\mu=1/2$, and $p=3/4$ throughout. Each entry is the gain from the report over acting on the prior, $J_A(V)-V(\mu)$. The renewal gains are in units of $h-\ell$. The forecast gains are reductions in expected squared error.
\end{minipage}
\end{table}

Table~\ref{tab:pilot-values} compares each report with acting on the prior. This matters for the second renewal buyer. She would already hire the contractor, so we subtract $1/6$ from each of her normalized payoffs in Section~\ref{sec:renewal}. Of the three decisions in the table, hers is the one for which the buyer prefers separate projects.

\section{Auditing a two-task project}\label{sec:audit}

Return to a buyer with a single two-record order. Suppose the qualification service opens a random fraction $\rho$ of its orders and discloses each record's result on those it opens. The buyer places the order before the audit draw. She learns whether an audit occurred when she receives the report, and sees only the usual verdict on unaudited orders. The draw is independent of competence and outcomes, has the same probability under both assignments, and has no direct effect on her continuation payoff. How often must the service check the records for the bundled project to dominate the narrow project?
Let $\bar p_2=(h+\ell)/(1+h\ell)$ be the unaudited threshold from Proposition~\ref{prop:one}, and define
\begin{equation}
p_A(\rho)=
\begin{cases}
\displaystyle \bar p_2+\frac{\rho}{1-\rho}(1-\bar p_2),&0\leq\rho<\frac12,\\[8pt]
1,&\frac12\leq\rho\leq1.
\end{cases}
\label{eq:audit-frontier}
\end{equation}

\begin{proposition}\label{prop:audit}
Suppose competence has two levels, the project has two tasks, and the audit rule is the same under both assignments. The bundled project Blackwell dominates the narrow project if and only if $p\leq p_A(\rho)$. The experiments are incomparable otherwise. Under the same audit rule the narrow project never dominates the bundled project.
\end{proposition}

For $p\leq\bar p_2$, no audit is needed. For $p>\bar p_2$, the minimum rate is
\begin{equation}
\rho^*(p)=\frac{p-\bar p_2}{1+p-2\bar p_2}.
\label{eq:min-audit}
\end{equation}
The minimum rises from zero at $\bar p_2$ to one half at $p=1$. To see why a random half suffices, let the buyer take a bundled pass as evidence of one correct task. This contributes probability $q^2$ at accuracy $q$. If the order is audited and exactly one conversion is correct, let her count that outcome as well. It contributes $2\rho q(1-q)$. At $\rho=1/2$ the two probabilities add to $q$. She has recovered the outcome of one expert task in distribution and can draw the outside outcome and the audit flag independently to reproduce a narrow project. The construction works at every accuracy because the two expert tasks are independent and equally reliable conditional on type. Proposition~\ref{prop:audit} gives the exact frontier for two competence levels. Even there, a smaller audit rate fails at $p=1$.\footnote{The comparison is made when the buyer assigns the order. The garbling can therefore draw a fresh audit flag for the simulated report. A comparison conditional on each realized flag would still require $p\leq\bar p_2$ in the unaudited branch whenever $\rho<1$.}

\section{Conclusion}\label{sec:conclusion}

At $p=3/4$ the bundled verdict is worth $1/50$ to each renewal buyer in units of $h-\ell$, relative to her best choice before the trial. Suppose instead that one separate order passes and the other fails. Both buyers put probability $8/15$ on the more accurate routine. The buyer with the more dependable incumbent keeps that supplier. The other renews. The first will pay extra to bundle, the second to separate.

With two possible accuracies, a bundled verdict can reproduce the separate reports at sufficiently low outside reliability. At equal net trial returns any buyer can then choose to bundle. A third accuracy removes that guarantee for unaudited orders. Whether the buyer still wants the tasks kept together depends on which reports would persuade her to change suppliers.

\appendix
\section{Proofs}\label{app:proofs}

\begin{lemma}\label{lem:binary}
Suppose that an experiment about two states has two signals with distinct finite likelihood ratios $D<U$, measured as high-state probability divided by low-state probability. It Blackwell dominates a finite target experiment if and only if every target signal of positive probability in either state has likelihood ratio in $[D,U]$. A signal possible only in the high state has ratio $+\infty$.
\end{lemma}

At any interior prior, multiplying by the prior odds converts a likelihood ratio into posterior odds. With two source signals the lemma is thus a special case of the more-dispersed criterion in \citet{Wu2023}. We give the reporting rule explicitly.

\begin{proof}
Consider a target signal $y$ produced by garbling. Its ratio is an average of the two source ratios, with each weight proportional to the low-state probability of receiving the corresponding source signal and reporting $y$. Moreover, both source signals occur in the low state, so a signal possible only in the high state cannot be generated.

Conversely, let the source assign probabilities $(a,1-a)$ to the two signals in the low state and $(Ua,D(1-a))$ in the high state. Omit null target signals. For each remaining $y$, write $b_y$ for its low-state probability and $r_y\in[D,U]$ for its likelihood ratio. After the source signals with ratios $U$ and $D$, respectively, report $y$ with probabilities
\[
x_y=\frac{b_y(r_y-D)}{a(U-D)},\qquad
z_y=\frac{b_y(U-r_y)}{(1-a)(U-D)}.
\]
Both probabilities are nonnegative. Since $\sum_y b_y=\sum_y b_y r_y=1$ and $aU+(1-a)D=1$, each row sums to one. The low- and high-state probabilities of $y$ are then $b_y$ and $r_yb_y$, as required.
\end{proof}

\begin{proof}[Proof of Proposition~\ref{prop:one}]
If $p=0$ the narrow project yields a certain failure, while the bundled verdict is informative. Now let $p>0$. The likelihood ratios in the bundled project run from $d_n$ to $(h/\ell)^n$. Those in the narrow project run from $t(p)=(1-ph)/(1-p\ell)$ to $h/\ell$. The upper ratio in the narrow project is strictly smaller, so that project can never reproduce a bundled success. By Lemma~\ref{lem:binary}, the bundled verdict can reproduce the narrow one if and only if $t(p)\geq d_n$. Solving for $p$ gives
\[
p\leq\frac{1-d_n}{h-\ell d_n}=\bar p_n.
\]
Above this cutoff neither reporting rule exists. Setting $n=2$ gives $\bar p_2=(h+\ell)/(1+h\ell)$.
\end{proof}

\begin{proof}[Proof of Proposition~\ref{prop:fixed}]
At $p=0$ the bundled project strictly dominates separate projects, which give the same count at every accuracy. For $p>0$, the probabilities of count $k$ under the high and low types are, respectively,
\[
\binom nk(ph)^k(1-ph)^{n-k}
\quad\text{and}\quad
\binom nk(p\ell)^k(1-p\ell)^{n-k}.
\]
Their ratio, given in \eqref{eq:fragment-lr}, increases strictly with $k$ since $h/\ell>1>t(p)$. At $k=n$ it is $(h/\ell)^n$, the ratio after a bundled success. At $k=0$ it is $t(p)^n$. By Lemma~\ref{lem:binary}, a garbling from the bundled verdict to the count exists if and only if
\[
t(p)^n\geq d_n,
\]
or $p\leq\widehat p_n$.

To reproduce the bundled verdict, a reporting rule must generate a success entirely from $k=n$. Giving positive weight to any other count would lower its ratio. In the low state the most probability that can be assigned to a pass is therefore $p^n\ell^n$. This is strictly less than the required $\ell^n$ for $0<p<1$. At $p=1$, announcing a pass just when $k=n$ reproduces the bundled project from separate projects.
\end{proof}

\begin{proof}[Proof of Corollary~\ref{cor:thresholds}]
Write the cutoffs as $\bar p_n=g(d_n)$ and $\widehat p_n=g(s_n)$, where
\[
g(x)=\frac{1-x}{h-\ell x}
\]
has derivative $(\ell-h)/(h-\ell x)^2<0$. Since $0<d_n<s_n=d_n^{1/n}<1$, both cutoffs are positive and $\widehat p_n<\bar p_n$. To bound them above, observe that
\[
d_n=\frac{1-h}{1-\ell}
\frac{1+h+\cdots+h^{n-1}}{1+\ell+\cdots+\ell^{n-1}}
>\frac{1-h}{1-\ell},
\]
and $g((1-h)/(1-\ell))=1$. Thus $\bar p_n<1$.

For monotonicity, integrate the derivative of $\log(1-x^n)$ to obtain
\[
-\log s_n=\int_{\ell}^{h}\frac{x^{n-1}}{1-x^n}\,dx.
\]
The integrand falls strictly with $n$ for every $x\in(\ell,h)$. Thus $s_n$ rises, and $\widehat p_n=g(s_n)$ falls strictly. Also,
\[
\log d_n=\log(1-h^n)-\log(1-\ell^n)=-h^n[1+o(1)].
\]
It follows that $1-s_n\sim h^n/n$ and $1-d_n\sim h^n$. The denominator of $g$ tends to $h-\ell$ in both cases, giving the asserted asymptotics. Taking $(n-1)$st roots expresses both cutoffs in per-task units. In particular,
\[
(n-1)\log\frac{\widehat r_n}{h}
=\log\frac{h}{h-\ell}-\log n+o(1).
\]
For $\bar r_n$ the same expression holds with the $-\log n$ term removed. Division by $n-1$ gives convergence to $h$ for both per-task cutoffs.

It remains to locate $h$ for finite $n$. At $p=h^{n-1}$ we have $h^{n-1}\ell>\ell^n$, so $t(p)=(1-h^n)/(1-h^{n-1}\ell)>d_n$. Two applications of the weighted arithmetic--geometric mean inequality give
\[
(1-h^n)^{(n-1)/n}(1-\ell^n)^{1/n}
<1-\frac{(n-1)h^n+\ell^n}{n}
<1-h^{n-1}\ell.
\]
Both inequalities are strict since $\ell<h$. Rearranging gives $t(p)^n<d_n<t(p)$. As $t$ decreases strictly, $t(\widehat p_n)^n=t(\bar p_n)=d_n$ places $h^{n-1}$ strictly between $\widehat p_n$ and $\bar p_n$.
\end{proof}

\begin{proof}[Proof of Proposition~\ref{prop:value}]
For $p>0$, Bayes' rule gives the posterior after $k$ separate successes. Cancelling $\binom nk p^k$ yields
\[
m_i(k;p)=
\frac{\pi_i q_i^k(1-pq_i)^{n-k}}
     {\sum_j \pi_j q_j^k(1-pq_j)^{n-k}} .
\]
The right-hand side stays in the interior of the simplex and extends continuously to $p=0$. For $k>0$ its value there is a limit at a count of zero probability. A finite convex function is continuous on the relative interior of its domain. Multiplying $V(m(k;p))$ by the probability of $k$ and summing over the finitely many counts gives continuity of $J_{\Separate_{n,p}}(V)$ on $[0,1]$.

Take $p_1<p_2$ and thin successes as in Section~\ref{sec:convex}. A thinned count of zero can come from an original count of either zero or one, each with positive probability. For $q_i<q_j$, observing one success instead of zero multiplies the posterior odds on $j$ against $i$ by
\[
\frac{q_j(1-p_2q_i)}{q_i(1-p_2q_j)}>1 .
\]
The thinned posterior therefore averages distinct original posteriors on an event of positive probability. Strict Jensen's inequality gives $J_{\Separate_{n,p_1}}(V)<J_{\Separate_{n,p_2}}(V)$.

At $p=0$ the separate projects give only the prior, while the bundled verdict gives two distinct posteriors because $q_i^n$ varies across types. Strict convexity gives $J_{\Separate_{n,0}}(V)<J_{\Bundled_n}(V)$. At $p=1$ a bundled failure pools at least the counts zero and one, with distinct posteriors by the odds calculation above. Hence $J_{\Bundled_n}(V)<J_{\Separate_{n,1}}(V)$. Continuity and strict monotonicity yield exactly one crossing. With two types, Proposition~\ref{prop:fixed} gives a garbling from the bundled project to separate projects at $p=\widehat p_n$. Each intermediate count has positive probability and a posterior strictly between the two bundled posteriors. The garbling mixes those posteriors with positive weights. Strict Jensen gives $J_{\Bundled_n}(V)>J_{\Separate_{n,\widehat p_n}}(V)$ and thus $p_V>\widehat p_n$.

For $d\geq3$, keep the prior fixed and consider forecasting a state-dependent payoff $a_i$. We choose the payoffs so that
\[
\sum_i\pi_i a_i=\sum_i\pi_i a_i q_i^n=0,
\qquad
\sum_i\pi_i a_i q_i^{n-1}>0.
\]
To see that this is possible, project $(q_i^{n-1})_i$ onto the span of $(1)_i$ and $(q_i^n)_i$ using the prior weights, and take $a$ to be the residual. Since $x^{(n-1)/n}$ is strictly concave, $q_i^{n-1}$ cannot be affine in $q_i^n$ at three distinct accuracies. The residual is nonzero, and its inner product with $(q_i^{n-1})_i$ is its squared norm.

Both bundled verdicts give posterior mean zero for $a_i$. For every $p>0$ there is a positive probability of exactly $n-1$ separate successes. The likelihood of that count is proportional across states to $q_i^{n-1}-p q_i^n$. The posterior mean after that count is strictly positive. Under squared-error loss the optimized payoff is $V_0(m)-\sum_i m_i a_i^2$, where $V_0(m)=(\sum_i a_i m_i)^2$. The affine term has the same expectation in the two assignments. Separate projects therefore give a strictly better forecast at every $p>0$. Fix any such $p$ and set $V_\varepsilon(m)=V_0(m)+\varepsilon\sum_i m_i^2$. For all sufficiently small $\varepsilon>0$ the preference is still strict. Each $V_\varepsilon$ is strictly convex and has a positive cutoff below the chosen $p$, proving that the infimum is zero.
\end{proof}

\begin{proof}[Derivation of \eqref{eq:quadratic-bundle}--\eqref{eq:quadratic-fragment}]
For any finite report, write $a_H(y)$ and $a_L(y)$ for the probabilities of signal $y$ conditional on type. Its probability under the prior is
\[
m(y)=\mu a_H(y)+(1-\mu)a_L(y).
\]
By Bayes' rule, at each signal of positive probability the buyer revises her forecast by
\[
M(y)-\mu
=\frac{\mu(1-\mu)[a_H(y)-a_L(y)]}{m(y)}.
\]
Squaring and averaging over those signals gives
\[
\E[(M-\mu)^2]
=\mu^2(1-\mu)^2
\sum_y\frac{[a_H(y)-a_L(y)]^2}{m(y)}.
\]
In the bundled project the squared difference in the numerator is the same for success and failure. Using $1/m_B+1/(1-m_B)=1/[m_B(1-m_B)]$ gives \eqref{eq:quadratic-bundle}. With separate projects the two conditional laws are binomial, yielding \eqref{eq:quadratic-fragment}.
\end{proof}

\begin{proof}[Proof of Proposition~\ref{prop:audit}]
For this proof let $R=h/\ell$ and $\kappa=\rho+(1-\rho)p$. At $p=\rho=0$ the narrow project always yields an unaudited failure, and the bundled verdict is informative. Otherwise $\kappa>0$. First reduce the narrow project to three possible reports. Let $H$ stand for an audited expert success or an unaudited project success, $A$ for an audited expert failure, and $U$ for an unaudited failure. At accuracy $q$ these have probabilities
\[
\kappa q,\qquad \rho(1-q),\qquad (1-\rho)(1-pq).
\]
These labels retain all the information in the original report. Given $H$, draw an audit with probability $\rho/\kappa$. On an audited report, the expert's result is known from the label and the outside result can be drawn with success probability $p$. The unaudited labels already specify the verdict. We can omit labels that have probability zero.

Under the same audit rule, pool bundled outcomes with equal likelihood ratios as in the table below. To recover the original labels after a double success, draw the audit flag with probability $\rho$. After an audited single success, either task was the successful one with probability $1/2$ at both accuracies.

\begin{center}
\small
\begin{tabular}{@{}lll@{}}
\toprule
Signal & Probability given $q$ & Likelihood ratio \\
\midrule
Double success & $q^2$ & $R^2$ \\
Audited, exactly one success & $2\rho q(1-q)$ & $h(1-h)/[\ell(1-\ell)]$ \\
Unaudited failure & $(1-\rho)(1-q^2)$ & $(1-h^2)/(1-\ell^2)$ \\
Audited double failure & $\rho(1-q)^2$ & $[(1-h)/(1-\ell)]^2$ \\
\bottomrule
\end{tabular}
\end{center}
Omitting rows of zero probability, the order of the likelihood ratios relative to $R$ is given by
\[
R^2>R>
\frac{h(1-h)}{\ell(1-\ell)}>
\frac{1-h^2}{1-\ell^2}>
\left(\frac{1-h}{1-\ell}\right)^2.
\]
How much probability can a reporting rule assign to $H$ at ratio $R$? Let $L_i$ be the probability of bundled row $i$ in the low state and $r_i$ its likelihood ratio. If a rule sends a fraction $x_i$ of each row to $H$ and achieves probability $a q$ at each accuracy, then $\sum_i L_i x_i=a\ell$ and $\sum_i L_i r_i x_i=ah$. Consequently, for every $u<R$,
\[
a\ell(R-u)=\sum_i L_i(r_i-u)x_i
\leq\sum_{i:\,r_i>u}L_i(r_i-u).
\]
Taking every row above $u$ in full and omitting every row below it gives equality.

For $\rho<1/2$, take $u$ equal to the ratio for unaudited failure. Assign all double successes and audited single successes to $H$, together with a fraction
\[
c=\frac{(1-2\rho)h\ell}{(1-\rho)(1+h\ell)}
\]
of unaudited failures. Write $a_{\max}=\bar p_2+2\rho(1-\bar p_2)$. Subtracting $a_{\max}q$ from the resulting probability gives $(1-2\rho)(q-h)(q-\ell)/(1+h\ell)$, which vanishes at both types. This rule therefore attains the bound. For $\rho\geq1/2$, choose $u$ equal to the audited single-success ratio and take all double successes and a fraction $1/(2\rho)$ of audited single successes. This yields $q$ and attains the bound with $a_{\max}=1$. Reproducing the narrow report requires $\kappa\leq a_{\max}$.

To complete the construction when $\rho<1/2$, start at $\kappa=a_{\max}$ so that the rule above supplies $H$ with the required probability. Send each unaudited failure to $A$ with probability $d=\rho/[2(1-\rho)]$, in addition to sending it to $H$ with probability $c$. Send each audited double failure to $A$ with probability $1/2$. Assign all remaining probability to $U$. These are valid probabilities since
\[
c<\frac{1-2\rho}{2(1-\rho)},
\qquad c+d<\frac12,
\]
and the probability assigned to $A$ is
\[
d(1-\rho)(1-q^2)+\frac{\rho}{2}(1-q)^2=\rho(1-q).
\]
The residual has the required probability for $U$. At a smaller $\kappa$, keep a generated $H$ with probability $\kappa/a_{\max}$ and send it to $U$ otherwise. The probability of $A$ is unaffected.

For $\rho\geq1/2$ the construction already gives an expert success indicator $X$ of probability $q$. Independently draw an outside success $Z$ with probability $p$ and an audit flag with probability $\rho$. Report the flag and either $(X,Z)$ on an audit or $XZ$ otherwise. This is the full narrow report at every accuracy. Solving $\kappa\leq a_{\max}$ yields \eqref{eq:audit-frontier}, with equality included. Finally, the narrow report has no likelihood ratio above $R$, and a bundled double success has positive probability and ratio $R^2$. The narrow project can never reproduce the bundled project.

\end{proof}

\end{document}